\documentclass[11pt]{article}

\usepackage[margin=1in]{geometry}
\usepackage[T1]{fontenc}
\usepackage{lmodern}
\usepackage{microtype}
\usepackage{amsmath,amssymb,amsthm}
\usepackage{mathtools}
\usepackage{algorithm}
\usepackage[noend]{algpseudocode}
\usepackage{booktabs}
\usepackage{enumitem}
\usepackage{xcolor}
\usepackage[colorlinks=true,linkcolor=blue!50!black,citecolor=green!40!black,urlcolor=blue!50!black]{hyperref}

\newtheorem{theorem}{Theorem}
\newtheorem{lemma}[theorem]{Lemma}
\newtheorem{proposition}[theorem]{Proposition}
\newtheorem{corollary}[theorem]{Corollary}
\theoremstyle{definition}

\newtheorem{fact}[theorem]{Fact}
\theoremstyle{remark}
\newtheorem{remark}[theorem]{Remark}

\newcommand{\good}{\textsf{G}}          
\newcommand{\bad}{\textsf{B}}           
\newcommand{\Test}{\textsc{Test}}
\newcommand{\BCE}{\textsc{BCE}}
\newcommand{\Retained}{C}               
\newcommand{\Disc}{D}                   
\newcommand{\nil}{\bot}
\newcommand{\gd}{\mathrm{good}}
\newcommand{\bd}{\mathrm{bad}}

\algrenewcommand\algorithmicrequire{\textbf{Input:}}
\algrenewcommand\algorithmicensure{\textbf{Output:}}

\title{Backtracking Candidate Elimination:\\
A One-Pass Algorithm for the Chip Testing Problem}

\author{Shiyi Chen\\
University of California, Berkeley\\
\texttt{yevette\_chen@berkeley.edu}
}

\date{}

\begin{document}

\maketitle

\begin{abstract}
In the \emph{chip testing problem}, we are given $n$ chips, strictly more than
half of which are good. Chips can test one another in pairs; a good chip always
reports the status of the other chip correctly, whereas a bad chip may report
arbitrarily and adversarially. The goal is to identify a single chip that is
guaranteed to be good. The problem originates in system-level fault diagnosis
and is closely related to the ``knights and spies'' puzzle. The standard textbook
solution is a halving recursion that tests disjoint pairs in rounds and keeps
one chip from each consistent pair.

We present the \emph{Backtracking Candidate Elimination} (\BCE) algorithm, a
sequential alternative that scans the chips once while maintaining a current
candidate and a stack of retained chips. Every chip is tested at most once as
the incoming chip; when a test is inconclusive the candidate and the incoming
chip are discarded together, and the algorithm backtracks to the most recently
retained chip. \BCE{} uses at most $n-1$ tests and $O(n)$ time, needs no
parity case analysis, and works online. Its correctness follows from two
invariants: the retained chips all have the same type, and every discarded
pair contains at least one bad chip. We explain how \BCE{} can be viewed as
the Boyer--Moore majority vote algorithm with its counter replaced by a stack of
physical witnesses, and why that replacement is needed. We also give an early
termination rule and a variant for the weaker model of one-directional tests.
\end{abstract}


\section{Introduction}\label{sec:intro}

Consider a set of $n$ supposedly identical integrated-circuit chips that can
test one another. A test places two chips $x$ and $y$ in a jig; each chip
reports whether the other is good or bad. A good chip always reports
correctly. A bad chip may report anything, and we assume the bad chips
collude adversarially. Given the promise that strictly more than $n/2$ of the
chips are good, how many tests are needed to find one chip that is certainly
good?

This is the \emph{chip testing problem}, popularized by Cormen, Leiserson,
Rivest, and Stein as Problem~4-5 of \emph{Introduction to
Algorithms}~\cite{CLRS09}. It is a clean special case of \emph{adaptive
system-level fault diagnosis}, a model introduced by Preparata, Metze, and
Chien~\cite{PMC67} and studied adaptively by Hakimi and
Nakajima~\cite{HN84}, Hakimi and Schmeichel~\cite{HS84}, and many
others~\cite{SHOS90,PU98}. In that setting, units of a multiprocessor system
test one another, and a faulty unit's verdicts are unreliable. A single good
unit is the crucial resource. Once one is known, it can diagnose everything
else. For this reason the ``find one good chip'' subproblem is often simply
called \emph{the chip problem}~\cite{ACRS04,ACRS06}. Closely related problems
appear as logic puzzles about truth-tellers and liars~\cite{Smu78,Ble83},
most recently in the form of \emph{knights} (who always tell the truth) and
\emph{spies} (who may lie or tell the truth as they wish)~\cite{Wil10,Wil16}.
The strict-majority promise is necessary. If at least half the chips are
bad, the bad chips can imitate the good ones perfectly, and no strategy can
succeed~\cite{PMC67,CLRS09}.

\paragraph{The standard solution.}
The textbook approach~\cite{CLRS09} is a halving recursion. Pair up the chips
and test each pair. If both chips in a pair report ``good,'' keep one of them.
Otherwise, discard both. With a parity rule for the leftover chip when the
count is odd, the surviving chips still have a strict good majority. That
yields the recurrence $T(n) = T(\lceil n/2\rceil) + \lfloor n/2\rfloor$ and
hence $\Theta(n)$ tests. The argument is elegant, but it is organized in
\emph{rounds}. The reader must reason about a whole level of pairs at once,
handle odd sizes, and apply induction on the recursion.

\paragraph{Our contribution.}
We present the \emph{Backtracking Candidate Elimination} (\BCE) algorithm, a
one-pass alternative based on the same elimination principle but using a
different representation of the state. Rather than a tree of rounds, \BCE{}
keeps a \emph{current candidate} $p$ and a stack $R$ of \emph{retained}
chips, and it processes the remaining chips one at a time:
\begin{itemize}[nosep]
  \item if $p$ and the incoming chip $s$ both call each other good, $s$ is
        retained (pushed onto $R$) and $p$ is kept;
  \item otherwise, $p$ and $s$ are discarded as a pair, and the algorithm
        \emph{backtracks}: the most recently retained chip becomes the new
        candidate.
\end{itemize}
The key point is that a bad chip can be retained only temporarily. We never
need to determine whether a retained chip is lying. If the retained chips
are bad, the scan will eventually eliminate every one of them. Our results
are as follows.
\begin{enumerate}[label=(\roman*),nosep]
  \item \BCE{} returns a good chip using at most $n-1$ tests, $O(n)$ time,
        and a single left-to-right pass. It can therefore run \emph{online},
        with chips arriving one at a time
        (Theorems~\ref{thm:correct} and~\ref{thm:cost}).
  \item The correctness proof is non-recursive and rests on two short
        invariants: the retained chips form a \emph{homogeneous} set
        (Lemma~\ref{lem:homog}), and discarded chips are \emph{bad-heavy}
        (Lemma~\ref{lem:discard}). No parity case analysis is required.
  \item We make precise the connection with the Boyer--Moore majority vote
        algorithm \textsc{Mjrty}~\cite{BM91}. \BCE{} is \textsc{Mjrty} with
        the counter replaced by a stack of witnesses. We explain why chip
        testing makes this replacement necessary (Section~\ref{sec:mjrty}).
  \item We give an early termination rule
        (Proposition~\ref{prop:early}). We also give a variant for the
        weaker model in which a test yields only \emph{one} chip's verdict
        on the other. In that variant, the stack order genuinely matters
        (Section~\ref{sec:oneway}).
\end{enumerate}
In keeping with the spirit of simplicity in algorithms, we do not claim that
the linear bound is new, and sequential strategies of this flavor may well
be folklore. Our aim is a presentation of the chip problem that
is short, sequential, and easy to teach, together with a proof that fits in
a paragraph.

\paragraph{Organization.}
Section~\ref{sec:prelim} fixes the model. Section~\ref{sec:dc} recalls the
standard halving algorithm. Section~\ref{sec:bce} presents \BCE{} with an
example, and Section~\ref{sec:analysis} proves its correctness and cost.
Section~\ref{sec:discussion} compares the two algorithms, relates \BCE{} to
majority voting, and discusses extensions. Section~\ref{sec:related} surveys
related work.

\section{Model and Basic Facts}\label{sec:prelim}

Let $X = \{x_1,\dots,x_n\}$ be a set of chips. Each chip has a hidden
\emph{type}, either good ($\good$) or bad ($\bad$). Let $g$ and $b$ denote
the numbers of good and bad chips, so $g + b = n$. We are promised that
\begin{equation}\label{eq:majority}
  g > b .
\end{equation}
For a set $S \subseteq X$ we write $\gd(S)$ and $\bd(S)$ for the number of
good and bad chips in $S$.

\paragraph{Tests.}
A \emph{test} $\Test(x,y)$ on two distinct chips returns an ordered pair of
verdicts $(v_x, v_y) \in \{\good,\bad\}^2$. Here $v_x$ is $x$'s claim about
$y$, and $v_y$ is $y$'s claim about $x$. A good chip's verdict is always
correct. A bad chip's verdict is arbitrary and may be chosen by an adaptive
adversary that knows all types, the algorithm, and the history of tests. We
say a test is \emph{consistent} if it returns $(\good,\good)$ and
\emph{inconsistent} otherwise. Table~\ref{tab:outcomes} lists the possible
outcomes.

\begin{table}[h]
\centering
\begin{tabular}{@{}ccl@{}}
\toprule
Types of $(x,y)$ & Possible verdicts $(v_x,v_y)$ & Consistent possible? \\
\midrule
$(\good,\good)$ & $(\good,\good)$ only                     & yes (forced) \\
$(\good,\bad)$  & $(\bad,\good)$ or $(\bad,\bad)$          & no \\
$(\bad,\good)$  & $(\good,\bad)$ or $(\bad,\bad)$          & no \\
$(\bad,\bad)$   & any of the four                          & yes \\
\bottomrule
\end{tabular}
\caption{Outcomes of $\Test(x,y)$ as a function of the chip types.}
\label{tab:outcomes}
\end{table}

The whole analysis rests on two immediate consequences of
Table~\ref{tab:outcomes}.

\begin{fact}\label{fact:consistent}
If $\Test(x,y)$ is consistent, then $x$ and $y$ have the same type.
\end{fact}

\begin{fact}\label{fact:inconsistent}
If $\Test(x,y)$ is inconsistent, then at least one of $x,y$ is bad.
Consequently, $\bd(\{x,y\}) \ge \gd(\{x,y\})$.
\end{fact}

The cost of an algorithm is the number of tests it performs in the worst
case over all type assignments satisfying~\eqref{eq:majority} and all
adversarial answers.

\section{The Standard Halving Algorithm}\label{sec:dc}

For comparison, we recall the textbook solution~\cite{CLRS09}. It is usually
called divide-and-conquer, though it is more precisely a prune-and-search
recursion. Partition the chips into $\lfloor n/2 \rfloor$ disjoint pairs,
leaving one chip $z$ aside if $n$ is odd, and test every pair. By
Fact~\ref{fact:inconsistent}, discarding both chips of every inconsistent pair
removes at least as many bad chips as good ones. By
Fact~\ref{fact:consistent}, each consistent pair is homogeneous, so keeping
exactly one chip from each consistent pair halves the good and bad counts
among those pairs. If $n$ is odd, $z$ is kept exactly when the number of
consistent pairs is even.\footnote{Let $a$ and $c$ be the numbers of
consistent pairs of good and of bad chips. The majority among the undiscarded
chips gives $2a + [z\in\good] > 2c + [z\in\bad]$. If $a+c$ is even, then
either $a=c$ and $z$ is good, or $a \ge c+2$; in both cases keeping $z$ is
safe. If $a+c$ is odd, then $a \ge c+1$ and dropping $z$ is safe.} A short
calculation shows that the surviving set, of size at most $\lceil n/2\rceil$,
still satisfies~\eqref{eq:majority}. Recursing until one chip remains gives
\[
  T(1) = 0,\qquad T(n) \le T(\lceil n/2 \rceil) + \lfloor n/2 \rfloor,
\]
so $T(n) \le n-1$.

The algorithm has a natural parallel structure: all tests within a round are
independent, and there are $\lceil \log_2 n\rceil$ rounds. However, its
correctness argument is inherently \emph{level-by-level}. It needs the whole
input up front, a parity rule for the odd chip, and an induction over rounds.
The question that motivates this note is whether the same elimination
principle can be applied \emph{one chip at a time}.

\section{The Backtracking Candidate Elimination Algorithm}\label{sec:bce}

\BCE{} fixes an arbitrary order $x_1,\dots,x_n$ of the chips and scans it
once. At all times it maintains
\begin{itemize}[nosep]
  \item a \emph{candidate} $p \in X \cup \{\nil\}$, and
  \item a stack $R$ of \emph{retained} chips, with $R$ empty whenever
        $p = \nil$.
\end{itemize}
We call $\Retained = \{p\} \cup R$ (or $\emptyset$ if $p=\nil$) the
\emph{retained set}. Chips that leave the retained set, and incoming chips
that never enter it, are \emph{discarded}. Initially $p = x_1$ and $R$ is
empty. When chip $x_i$ arrives:

\begin{enumerate}[label=\textbf{Case \arabic*.},leftmargin=*]
  \item \textbf{No candidate ($p=\nil$).} Make $x_i$ the candidate without
        testing it.
  \item \textbf{Consistent test.} If $\Test(p,x_i)$ returns
        $(\good,\good)$, then $p$ and $x_i$ have the same type, although we
        do not know which. Retain $x_i$ by pushing it onto $R$, and keep
        $p$ as the candidate.
  \item \textbf{Inconsistent test.} Otherwise, $\{p,x_i\}$ contains a bad
        chip, and both are discarded. The algorithm then \emph{backtracks}.
        If $R$ is nonempty, the most recently retained chip is popped and
        becomes the new candidate. If $R$ is empty, set $p = \nil$.
\end{enumerate}

At the end of the scan, \BCE{} returns $p$. Pseudocode appears as
Algorithm~\ref{alg:bce}.

\begin{algorithm}[t]
\caption{Backtracking Candidate Elimination (\BCE)}\label{alg:bce}
\begin{algorithmic}[1]
\Require chips $x_1,\dots,x_n$ with strictly more good than bad chips
\Ensure a good chip
\State $p \gets x_1$;\quad $R \gets$ empty stack
\For{$i = 2,\dots,n$}
  \If{$p = \nil$}
     \State $p \gets x_i$ \Comment{Case 1: restart without a test}
  \ElsIf{$\Test(p, x_i) = (\good,\good)$}
     \State $\textsc{Push}(R, x_i)$ \Comment{Case 2: retain $x_i$, keep $p$}
  \Else \Comment{Case 3: discard $p$ and $x_i$, then backtrack}
     \If{$R$ is nonempty}
        \State $p \gets \textsc{Pop}(R)$
     \Else
        \State $p \gets \nil$
     \EndIf
  \EndIf
\EndFor
\State \Return $p$
\end{algorithmic}
\end{algorithm}

\paragraph{Example.}
Suppose the hidden types, in scan order, are
\[
  x_1,\dots,x_7 \;=\; \good,\ \good,\ \bad,\ \bad,\ \bad,\ \good,\ \good .
\]
Table~\ref{tab:trace} traces the execution. The trace does not depend on how
the bad chips answer. Every test involves at least one good chip, whose
verdict is forced. The pairs $\{x_1,x_3\}$, $\{x_2,x_4\}$, and $\{x_5,x_6\}$
are discarded. The candidate backtracks from $x_1$ to the retained chip $x_2$,
and \BCE{} returns the good chip $x_7$ after four tests.

\begin{table}[h]
\centering
\begin{tabular}{@{}cccllc@{}}
\toprule
$i$ & type of $x_i$ & candidate $p$ & action & result & $R$ after \\
\midrule
2 & $\good$ & $x_1$     & $\Test(x_1,x_2) = (\good,\good)$ & retain $x_2$                          & $[x_2]$ \\
3 & $\bad$  & $x_1$     & $\Test(x_1,x_3)$ inconsistent    & discard $x_1,x_3$; $p\gets x_2$       & $[\,]$ \\
4 & $\bad$  & $x_2$     & $\Test(x_2,x_4)$ inconsistent    & discard $x_2,x_4$; $p\gets\nil$       & $[\,]$ \\
5 & $\bad$  & $\nil$    & no test                          & $p \gets x_5$                          & $[\,]$ \\
6 & $\good$ & $x_5$     & $\Test(x_5,x_6)$ inconsistent    & discard $x_5,x_6$; $p\gets\nil$       & $[\,]$ \\
7 & $\good$ & $\nil$    & no test                          & $p \gets x_7$                          & $[\,]$ \\
\bottomrule
\end{tabular}
\caption{Execution of \BCE{} on the types $\good\good\bad\bad\bad\good\good$.
The output is $x_7$.}
\label{tab:trace}
\end{table}

Two features of the example deserve emphasis. First, the bad chip $x_5$
temporarily became the candidate. \BCE{} never tries to detect that
$x_5$ is bad. It simply discards $x_5$ together with the good chip that
exposes an inconsistency. Second, $x_2$ was retained and later re-selected by
backtracking. Each chip is tested \emph{at most once as the incoming chip},
although it may later serve as the candidate in further tests.

\section{Analysis}\label{sec:analysis}

We analyze \BCE{} through two invariants that hold after each iteration of
the loop. Let $\Retained_i$ be the retained set and $\Disc_i$ the set of
discarded chips after chip $x_i$ has been processed, with
$\Retained_1 = \{x_1\}$ and $\Disc_1 = \emptyset$.

\begin{lemma}[Homogeneity]\label{lem:homog}
For every $i$, all chips in $\Retained_i$ have the same type.
\end{lemma}

\begin{proof}
By induction on $i$. The set $\Retained_1 = \{x_1\}$ is trivially
homogeneous. Suppose $\Retained_{i-1}$ is homogeneous. In Case~1,
$\Retained_i = \{x_i\}$. In Case~2, $\Retained_i = \Retained_{i-1} \cup
\{x_i\}$, and $x_i$ has the same type as $p \in \Retained_{i-1}$ by
Fact~\ref{fact:consistent}. In Case~3, $\Retained_i =
\Retained_{i-1}\setminus\{p\}$ is a subset of a homogeneous set. In every
case, $\Retained_i$ is homogeneous.
\end{proof}

\begin{lemma}[Discarded chips are bad-heavy]\label{lem:discard}
For every $i$, $\bd(\Disc_i) \ge \gd(\Disc_i)$.
\end{lemma}

\begin{proof}
Chips are discarded only in Case~3, and always two at a time, as a pair
$\{p, x_i\}$ that has just produced an inconsistent test. By
Fact~\ref{fact:inconsistent}, each such pair contains at least as many bad
chips as good ones. Summing over the disjoint discarded pairs proves the
claim.
\end{proof}

Since every chip $x_1,\dots,x_i$ is either retained or discarded after step
$i$, we have the partition
\begin{equation}\label{eq:partition}
  \{x_1,\dots,x_i\} = \Retained_i \,\sqcup\, \Disc_i .
\end{equation}

\begin{theorem}[Correctness]\label{thm:correct}
Under the promise $g > b$, \BCE{} returns a good chip.
\end{theorem}

\begin{proof}
Apply~\eqref{eq:partition} with $i=n$. Then
\[
  \gd(\Retained_n) - \bd(\Retained_n)
  = (g - b) - \bigl(\gd(\Disc_n) - \bd(\Disc_n)\bigr)
  \;\ge\; g - b \;>\; 0,
\]
where the inequality uses Lemma~\ref{lem:discard}. Thus $\Retained_n$
contains at least one good chip. In particular, $\Retained_n \ne \emptyset$,
so $p \ne \nil$. By Lemma~\ref{lem:homog}, every chip in $\Retained_n$ is
good, and in particular the returned candidate $p$ is good.
\end{proof}

The proof shows more than the theorem states. It captures precisely the
intuition that bad chips can be retained only temporarily.

\begin{corollary}[All retained bad chips are eventually eliminated]
\label{cor:no-bad}
At termination, the retained set $\Retained_n$ consists only of good chips.
Hence every bad chip that is ever retained or becomes the candidate is
eventually discarded, regardless of how the bad chips answer.
\end{corollary}

Concretely, a run of bad chips can survive for a while only if they are
retained under a bad candidate, since a good candidate exposes a bad chip
immediately. Because good chips are in the majority, the scan eventually
meets enough good chips to pair off and discard all of them, one backtracking
step at a time. Lemmas~\ref{lem:homog} and~\ref{lem:discard} turn this
informal ``re-examination'' argument into a two-line calculation.

\begin{theorem}[Cost]\label{thm:cost}
\BCE{} performs at most $n-1$ tests and runs in $O(n)$ time.
\end{theorem}

\begin{proof}
Each iteration $i = 2,\dots,n$ performs at most one test, namely
$\Test(p,x_i)$, and $O(1)$ stack operations.
\end{proof}

\begin{remark}[Space and online operation]\label{rem:space}
\BCE{} reads the chips in a single pass and never revisits a discarded chip.
It therefore applies unchanged when chips arrive one at a time and the
algorithm must, at any moment, be able to name its current candidate. Its
working memory is the retained set, which may contain $\Theta(n)$ chips when
many consecutive tests are consistent. Section~\ref{sec:mjrty} explains why
this is unavoidable for this style of algorithm in the chip model.
\end{remark}

\begin{remark}[Detecting a broken promise]
If \BCE{} ends with $p=\nil$, the proof of Theorem~\ref{thm:correct} shows
that the promise $g>b$ was violated. The converse does not hold. Indeed, no
algorithm can verify the promise, because bad chips can mimic good ones.
\end{remark}

\subsection{Early termination}

The partition~\eqref{eq:partition} also shows that \BCE{} can often stop
before reading the whole input.

\begin{proposition}[Early termination]\label{prop:early}
Suppose that after processing $x_i$ the retained set has size
$h = |\Retained_i|$, and $r = n - i$ chips remain unread. If $h \ge r$ and
$h \ge 1$, then every chip in $\Retained_i$ is good. In particular, \BCE{}
may stop and return $p$.
\end{proposition}

\begin{proof}
Suppose, for contradiction, that $\Retained_i$ is not all good. By
Lemma~\ref{lem:homog} it is then all bad, so $b \ge h + \bd(\Disc_i)$. The
good chips lie in $\Disc_i$ or among the $r$ unread chips, so
$g \le \gd(\Disc_i) + r \le \bd(\Disc_i) + r$ by Lemma~\ref{lem:discard}.
Combining these bounds with $g > b$ gives $r > h$, a contradiction.
\end{proof}

For example, if all chips are good, the retained set grows by one per test,
and \BCE{} stops after about $n/2$ tests instead of $n-1$. Independently, when
$n$ is even, the promise gives $g \ge b + 2$, so one chip may be dropped at
the start. This saves one test in the worst case.

\section{Discussion}\label{sec:discussion}

\subsection{Halving versus backtracking}

Both algorithms enforce the same majority invariant: discard only
bad-heavy sets. They organize the surviving information differently.
Table~\ref{tab:compare} summarizes the comparison.

\begin{table}[h]
\centering
\begin{tabular}{@{}lll@{}}
\toprule
 & Halving recursion~\cite{CLRS09} & \BCE{} (this paper) \\
\midrule
Structure              & rounds of disjoint pairs            & single scan with a stack \\
What a consistent test does & keep one chip, discard the other & keep both chips \\
Worst-case tests       & $\le n-1$                           & $\le n-1$ \\
Rounds / parallelism   & $\lceil\log_2 n\rceil$ parallel rounds & sequential \\
Input access           & all chips up front                  & online \\
Odd sizes              & parity rule needed                  & no special case \\
Proof                  & induction over rounds               & two invariants \\
Early termination      & not immediate                       & Proposition~\ref{prop:early} \\
\bottomrule
\end{tabular}
\caption{Comparison of the two linear-test algorithms for the chip problem.}
\label{tab:compare}
\end{table}

A conceptual difference is what happens after a consistent test. The halving
algorithm must \emph{throw away} one chip of each consistent pair to restore
the majority invariant at the next level. \BCE{} keeps both chips, and
homogeneity (Lemma~\ref{lem:homog}) accounts for them together. The
information from a consistent test is therefore never lost. It is stored in
the stack until a later inconsistent test uses it.

\subsection{Relation to Boyer--Moore majority voting}\label{sec:mjrty}

Readers familiar with the Boyer--Moore majority vote algorithm
\textsc{Mjrty}~\cite{BM91} (see also Misra and Gries~\cite{MG82}) will
recognize the pattern. \textsc{Mjrty} scans a sequence of votes, keeping a
candidate value and a counter. A vote equal to the candidate increments the
counter, and a different vote decrements it, implicitly cancelling a pair of
distinct votes. In \textsc{Mjrty} the ``retained set'' consists of $k$ copies
of the same value, so it can be stored as a (value, counter) pair.

\BCE{} is exactly this idea transplanted to chip testing, where
Lemma~\ref{lem:homog} plays the role of ``all retained votes are equal.'' The
difference is that chips are \emph{not values that can be copied}. A test
must involve two physical chips. Once the candidate is discarded after an
inconsistent test, the algorithm needs another \emph{actual chip} of the same
type to continue. A counter records how many such chips exist but not which
ones they are. \BCE{} therefore materializes the counter as the stack $R$ of
witnesses: the stack height is the \textsc{Mjrty} counter, and backtracking
is the decrement. This explains the $\Theta(n)$ worst-case memory in
Remark~\ref{rem:space}, compared with the $O(\log n)$ bits used by
\textsc{Mjrty}.

In the two-sided model of Section~\ref{sec:prelim}, the chips in $R$ are
interchangeable by Lemma~\ref{lem:homog}. Backtracking to the \emph{most
recent} retained chip is thus a convenient choice rather than a necessary
one, and any container would do. The next subsection shows a slightly weaker
model in which the last-in-first-out order is essential.

\subsection{One-directional tests}\label{sec:oneway}

In the classical diagnosis model of Preparata, Metze, and Chien~\cite{PMC67},
and in the knights-and-spies puzzle~\cite{Ble83,Wil10}, a test is
\emph{directed}. We ask chip $x$ about chip $y$ and receive a single verdict.
Such a test is weaker, because only one of the two verdicts is received. Nevertheless, a stack-based scan still finds a good chip with at most
$n-1$ directed tests.

Maintain the retained chips as a stack $c_1, c_2, \dots, c_h$, from bottom to
top. When $x_i$ arrives with a nonempty stack, ask $x_i$ about the top chip
$c_h$. If $x_i$ says $c_h$ is good, push $x_i$. Otherwise, pop $c_h$ and
discard both $c_h$ and $x_i$. If the stack is empty, push $x_i$ without a
question. Return the \emph{bottom} chip $c_1$.

The invariant replacing Lemma~\ref{lem:homog} is \emph{monotonicity}: no good
chip lies above a bad chip in the stack. If $x_i$ is good and vouches for
$c_h$, then $c_h$ is good, so pushing preserves monotonicity. If $x_i$ is
bad, it sits on top, which is also allowed. Popping preserves monotonicity
trivially. If $x_i$ accuses $c_h$, then $x_i$ and $c_h$ cannot both be good,
so discarded pairs are still bad-heavy. The calculation in
Theorem~\ref{thm:correct} shows that the final stack has a good majority. By
monotonicity, its bottom chip is good. Here the stack discipline matters:
the incoming chip must be compared against the \emph{top}, and the answer is
read from the \emph{bottom}. Chains of this kind, in which each person vouches
for the previous one, also underlie Wildon's analysis of the knights-and-spies
problem~\cite{Wil10}.

\subsection{Optimality}

\BCE{} is not designed to optimize lower-order terms. For the closely related
\emph{majority problem}, where one must find an element of the majority color
using equal/unequal comparisons, exactly $n - \nu(n)$ comparisons are
necessary and sufficient in the worst case. Here $\nu(n)$ denotes the number
of ones in the binary expansion of $n$~\cite{SW91,ARS93,Wie02}. Alonso,
Chassaing, Reingold, and Schott~\cite{ACRS04} showed that the chip problem is
closely related to a modified majority problem in the worst case and derived
upper and lower bounds from this connection. The optimal algorithms merge
homogeneous groups of equal power-of-two sizes, in the manner of binomial
heaps. Compared with such algorithms, \BCE{} gives up a lower-order saving of at
most logarithmically many tests in exchange for a single sequential pass and
a proof that fits in a paragraph.

\section{Related Work}\label{sec:related}

\paragraph{System-level fault diagnosis.}
Preparata, Metze, and Chien~\cite{PMC67} introduced the model in which units
test one another and faulty units give unreliable verdicts. They showed that
diagnosis requires fewer than half of the units to be faulty. Adaptive
diagnosis, where each test may depend on previous outcomes, was developed by
Hakimi and Nakajima~\cite{HN84} and Hakimi and Schmeichel~\cite{HS84}.
Parallel and probabilistic variants followed~\cite{SHOS90,PU98}. In adaptive
algorithms, a standard first phase finds one reliable unit, which is the
chip problem, and then uses it to diagnose the rest.

\paragraph{The chip problem.}
The two-sided pairwise formulation used here is the one of Cormen et
al.~\cite{CLRS09}, whose intended solution is the halving recursion of
Section~\ref{sec:dc}. Alonso et al.\ studied the worst case~\cite{ACRS04} and
the average case~\cite{ACRS06} of finding one good chip, building on
majority-problem techniques.

\paragraph{Knights, knaves, and spies.}
Puzzles about truth-tellers and liars go back at least to
Smullyan~\cite{Smu78}. Blecher~\cite{Ble83} studied a question-counting version of the
truth-tellers-and-liars puzzle using an adversary argument. Wildon~\cite{Wil10}
determined exactly the number of directed questions needed to identify all
knights and spies, via the Spider Interrogation Strategy. Its first phase is a
candidate-elimination procedure related to ours. In later work~\cite{Wil16},
he studied a majority/minority search game.

\paragraph{Majority voting and the majority problem.}
The Boyer--Moore algorithm~\cite{BM91} and its generalization by Misra and
Gries~\cite{MG82} are the canonical one-pass, pairwise-cancellation
algorithms. Fischer and Salzberg~\cite{FS82} gave tight bounds for determining
a majority with arbitrarily many colors. For two colors, the exact
worst-case bound $n-\nu(n)$ is due to Saks and Werman~\cite{SW91}, with
simpler proofs by Alonso, Reingold, and Schott~\cite{ARS93} and by
Wiener~\cite{Wie02}. The average case was settled in~\cite{ARS97}. Aigner's
survey~\cite{Aig04} covers many variants.

\section{Conclusion}\label{sec:conclusion}

We presented Backtracking Candidate Elimination, a one-pass algorithm that
finds a good chip among $n$ chips with a good majority using at most $n-1$
pairwise tests. The halving algorithm keeps the majority invariant by
compressing consistent pairs level by level. \BCE{} instead keeps a
homogeneous stack of retained chips and discards a bad-heavy pair whenever a
test is inconsistent. Chips retained under a bad candidate are never
explicitly identified as bad; the scan simply eliminates them later through
backtracking. The resulting proof reduces to a single counting identity. It
also exposes the problem as a physical version of Boyer--Moore majority
voting, and it extends to one-directional tests. We hope this viewpoint is
useful for teaching the chip problem alongside, or in place of, the
recursive solution.

\bibliographystyle{plainurl}
\bibliography{refs}

\end{document}